\documentclass[12pt, letter]{article}
\renewcommand{\baselinestretch}{1.5}
\usepackage{amsmath}
\usepackage{epsfig}
\usepackage{amssymb}
\usepackage{amscd}
\usepackage{palatino}
\usepackage{mathptmx} 
\usepackage{graphics}
\usepackage{enumerate}
\usepackage{hhline}
\usepackage{natbib}
\usepackage{pifont}
\usepackage{setspace}
\usepackage{etoolbox}
\usepackage{chngpage}
\usepackage{rotating}
\usepackage{multirow}
\usepackage{color}
\usepackage{colortbl}
\usepackage[separate]{patty}
\def\argmax{\mathop{\rm argmax}\limits}
\def\max{\mathop{\rm max}\limits}
\usepackage{authblk}
\usepackage{float}
\AtBeginEnvironment{tabular}{\doublespacing}

\begin{document}
\bibliographystyle{apsr}
\renewcommand{\baselinestretch}{1.5}

\title{Accuracy Is (Generically) Bad For Compliance\thanks{We thank Lauren Klein, Kevin Quinn, and Alexander Tolbert for very helpful conversations.}}

\author{John W. Patty\thanks{Professor of Political Science and Quantitative Theory \& Methods, Emory. Email: \textit{jwpatty@gmail.com}.} \, and Elizabeth Maggie Penn\thanks{Professor of Political Science and Quantitative Theory \& Methods, Emory. Email: \textit{elizabeth.m.penn@gmail.com}.}}

\date{\today}

\maketitle

\begin{abstract}
    We demonstrate that the set of cost distributions under which the optimal strategy for maximizing compliance (or more generally, effort) in a binary choice environment is identical to the optimal strategy for maximizing the accuracy of the reward (minimizing Type-I and Type-II errors) is \textit{finitely shy} (\cite{AndersonZame01}) in the space of all smooth parameterized real-valued distributions possessing full support on $\mathbf{R}$.  In words, this implies that maximizing compliance and maximizing accuracy ``almost always'' call for different incentive schemes.
\end{abstract}

It is well-known that there is a distinction between a principal attempting to use a material reward to provide maximal incentive for an agent to ``work hard'' versus attempting to reward the agent as accurately as possible in the sense of the agent getting a reward if and only if the agent actually worked hard.  However, there are also many ``workhorse'' models of principal-agent situations in which the two strategies coincide.

In many of these models, these optimal strategies coincide at least in part because of the timing of the principal's decision about whether to reward the agent.  For example, some of the situations captured by the (simple, but abstract) model we consider below, the principal's optimal strategy is identical under the two motivations \textit{because the principal makes the decision after the agent has already chosen his or her effort when the principal makes his or her reward decision}.  In such situations, the notion of \textit{sequential rationality} often implies that the principal should reward any agent whose efforts produced an informative signal implying the probability the agent exerted effort is sufficiently close to 1.  

Accordingly, we consider an alternative timing in this article: the principal can publicly commit to a reward strategy prior to the agent's choice of effort level.  It is important to note at this point that \textit{the principal always weakly benefits from having the ability to commit}.  Accordingly, while it is often an unrealistic assumption in practice, it is theoretically compelling, even if only because this setting increases the principal's ability to influence the agent's choice.\footnote{A timely and close analogy to this juxtaposition can be seen in the distinction between classical cheap-talk signaling models (\textit{e.g.}, \cite{CrawfordSobel82}), which do not allow the ``signaler'' to commit to a messaging strategy, and more recent models of Bayesian persuasion (\textit{e.g.}, \cite{KamenicaGentzkow11}).}  This is a short article with one main result: \textbf{incentivizing effort is generically distinct from accurately rewarding effort}.  With that, we now turn to the primitives of the analysis, which require more details than the proof of the main result itself.

\section{Technical Preliminaries}

\paragraph{Signal Distributions.} We refer to a pair of cumulative distribution functions (CDFs), $G \equiv (G_0,G_1)$, of real-valued random variables as \textbf{admissible} if $G_0$ and $G_1$ are each twice continuously differentiable, each possessing full support on $\mathbf{R}$, and the probability density functions (PDFs), $g_0$ and $g_1$, satisfy the monotone likelihood ratio property (MLRP):
\[
\frac{\dee \left[\frac{g_1(x)}{g_0(x)}\right]}{\dee x} > 0 \text{ for all } x \in \mathbf{R},
\]
and (without loss of generality) we normalize any admissible pair of distributions such that 
\begin{equation}
\label{Eq:Normalization}
g_0(t) = g_1(t) \Leftrightarrow t = 0.    
\end{equation}

\paragraph{Cost Distributions.}  We denote the CDF of a real-valued random variable, $c \in \mathbf{R}$, by $F$.  We assume throughout that $F$ is twice continuously differentiable with full support on $\mathbf{R}$.  We will impose more structure on this space for the main result, but the equilibrium analysis, to which we now turn, requires only that $F$ be twice continuously differentiable

\section{Equilibrium Analysis}

\cite{Penn25} establishes that, in our setting, $D$'s equilibrium strategy will be either a \textit{positive threshold strategy} or a \textit{negative threshold strategy}. Depending on whether the principal uses a positive or negative threshold, a threshold strategy with threshold $t \in \bar{\mathbf{R}}$\footnote{The notation $\bar{\mathbf{R}}$ represents the extended real line: $\{-\infty,+\infty\}\cup \mathbf{R}$.} will induce equilibrium prevalence equal to:
\begin{equation}
\label{Eq:EquilibriumPrevalences}
\begin{array}{ccl}
     \pi_F^+(t) = & F\left(r\cdot \left(G_0(t)-G_1(t)\right)\right) & \text{ for a positive threshold rule, and}\\
     \pi_F^-(t) = & F\left(r\cdot \left(G_1(t)-G_0(t)\right)\right) & \text{ for a negative threshold rule.}
\end{array}
\end{equation}
The following simple result is very useful.
\begin{lemma}
    \label{Lem:ComplianceMaximizationAndMinimization}
    If $F$ is a twice continuously differentiable CDF with full support on $\mathbf{R}$, then
    \[
    \argmax_{t\in \bar{\mathbf{R}}} \pi_F^+(t) = \argmin_{t\in \bar{\mathbf{R}}} \pi_F^-(t) = \left\{0\right\}.
    \]
\end{lemma}
\begin{proof}
    Follows by the normalization in \ref{Eq:Normalization} and because $G_1$ satisfies MLRP with respect to $G_0$.
\end{proof}
The following result is also useful because it allows us to focus only on positive threshold rules.
\begin{lemma}
    The optimal strategy for compliance maximization is always a positive threshold rule:
    \[
    \pi_F^+(t)>\pi_F^-(t) \text{ for all } t\in \mathbf{R}.
    \]
\end{lemma}
\begin{proof}
    Follows from standard arguments (\textit{e.g.}, \cite{JungEtAl20}, \cite{PennPatty24AlgEndog}).
\end{proof}
The following remark clarifies that positive and negative threshold rules yield the same compliance if and only if neither has a finite threshold.\footnote{\cite{PennPatty24AlgEndog} refer to such strategies as \textit{null}, because they do not affect individuals' behavioral incentives.}
\begin{remark}
    Note that the following are equivalent:
    \begin{enumerate}
        \item $\pi_F^+(t)=\pi_F^-(t)$,
        \item $\pi_F^+(t)=\min_{s\in \bar{\mathbf{R}}}\left[\pi_F^+(s)\right]$, and
        \item $\pi_F^-(t)=\max_{s\in \bar{\mathbf{R}}}\left[\pi_F^-(s)\right]$.
    \end{enumerate} 
    Each is true if and only if $t \in {-\infty,+\infty}$.
\end{remark}

\paragraph{Principal's Payoff.} Denote the principal's payoff from a positive threshold $t \in \bar{\mathbf{R}}$ by
\begin{equation*}
    \label{Eq:PositiveThresholdUtility}
    EU_F^+(t) = \pi_F^+(t) (1-G_1(t)) + (1-\pi_F^+(t)) G_0(t).
\end{equation*}
Taking the first derivative of $EU_F^+(t)$:
\begin{eqnarray*}
    \partial_t EU_F^+(t) 
    & = & \partial_t \pi_F^+(t) \cdot (1-G_0(t)-G_1(t)) - \pi_F^+(t) (g_0(t)+g_1(t)) + g_0(t).
\end{eqnarray*}
Note that $\pi_F^+(t)$ is uniquely maximized at $t=0$,\footnote{And, similarly, $\pi_F^-(t)$ is uniquely minimized at $t=0$.} so that 
\begin{eqnarray*}
    \partial_t EU_F^+(0) & = & (1-\pi_F^+(0)) g_0(0) - \pi_F^+(0) g_1(0).
\end{eqnarray*}
Furthermore, $g_0(0)=g_1(0)$ by construction, so that 
\begin{eqnarray*}
    \partial_t EU_F^+(0) & = & (1-2\pi_F^+(0)) g_0(0),
\end{eqnarray*}

\subsection{Topological Matters}  

Let $\mathcal{D}$ denote the set of continuous probability density functions on $\mathcal{R}$:
\[
\mathcal{D}(\mathbf{R}) \equiv \left\{f:\mathbf{R} \to \mathbf{R} \; \bigg\mid \; f \in \mathcal{C}(\mathbf{R}),f\geq 0,\int_{\mathbf{R}} f(x) \; \dee x = 1 \right\}.
\]
Unfortunately, $\mathcal{D}$ is not a topological vector space, because it is not closed under addition or scalar multiplication.  However, for our purposes, we can expand this space to the the space of Lebesgue integrable functions (\textit{i.e.}, the $L^1$ space):
\[
L^1(\mathbf{R}) \equiv \left\{f:\mathbf{R} \to \mathbf{R} \; \bigg\mid \; f \in \mathcal{C}(\mathbf{R}),f\geq 0,\int_{\mathbf{R}} |f(x)| \; \dee x < \infty \right\}.
\]
and endow this space with the topology induced by the $L^1$ norm:
\[
||f-h||_1 = \int_{\mathbf{R}} |f(x)-h(x)| \; \dee x.
\]
Endowed with the $L^1$ norm topology, $\mathcal{D}(\mathbf{R})$ is a convex Borel subset of $L^1$.

\paragraph{A Finite-Dimensional Space of Twice Continuously Differentiable Distributions on $\mathbf{R}$.} For any positive integer $k>0$, let $\mathcal{P}^k$ denote a family of random variables indexed by $x = (x_1,\ldots,x_k)\in X \subseteq \mathbf{R}^k$, where $X$ is a convex subset of $\mathbf{R}^k$,\footnote{We need to allow for $X\subset \mathbf{R}^k$ for distributional families such as the $\mathrm{Normal}(\mu,\sigma^2)$ distribution, where $\sigma>0$. By a ``family of admissible distributions,'' we mean that every pair of distributions in $X$ is an admissible pair.  We are unaware of situations in which this requirement is relevant, but that also suggests to us that it is not terribly important to explore how much this requirement constrains the applicability of the main result.} satisfying the following properties for each $x \in X$:
\begin{enumerate}
    \item \label{Pk1} (\textit{Twice-Smoothness.}) The CDF of $F \in \mathcal{P}^k$ is twice continuously differentiable:
    \[
    \begin{array}{r}
    f(z|x) \equiv \left.\frac{\dee F(t|x)}{\dee t}\right|_{t=z} \\
    \mbox{}\\
    \text{and } \;\;\; f'(z|x) \equiv \left.\frac{\dee f(t|x)}{\dee t}\right|_{t=z} 
    \end{array}
    \;\;\;\;\;\; \text{exist and are continuous in } z\in \mathbf{R}.
    \]
    \item \label{Pk2} (\textit{Linearity.}) The family of random variables is linear in $X$ in the sense that, for any $\alpha\in[0,1]$ and any $x,y \in \mathcal{P}^k$,
    \[
    f(\cdot|\alpha \cdot x+(1-\alpha) \cdot y) = \alpha \cdot f(\cdot|x) + (1-\alpha) \cdot f(\cdot|y), \text{ and }
    \]
    \item \label{Pk3} (\textit{Responsiveness to $x\in \mathcal{R}^k$.}) For any $\varepsilon>0$ and $x \in \mathcal{P}^k$, letting $B_\varepsilon(x)$ denote the $\epsilon$-open neighborhood of random variables centered on $x$:
    \[
    B_\varepsilon(x) \equiv \left\{ x' \in \mathcal{P}^k: ||x-x'|| < \epsilon \right\},
    \]
    letting $\lambda_{\mathbf{R}^k}$ denote $k$-dimensional Lebesgue measure, it is the case that 
    \[
    \lambda_{\mathbf{R}^k}\left(x' \in B_\varepsilon(x): F(t|x) = F(t|x') \right) = 0.
    \]
\end{enumerate}
Note that, when endowed with the $L^1$ norm topology, $\mathcal{P}^k$ is a finite dimensional convex Borel subset of $L^1$.  Note that Requirement \ref{Pk1} is essentially book-keeping in nature.  Requirement \ref{Pk2} merely imposes structure on the family that is true for many families of distributions used in the real world.  Somewhat similarly, Requirement \ref{Pk3} is essentially ruling out the possibility that the multi-index $x$ is irrelevant.  Put another (equivalent) way, a parameterized cost distribution $F$ is responsive if and only if changing all of its parameters by an arbitrary small amount changes the value of the CDF of the distribution at almost every value $t \in \mathbf{R}$.  To make clear that this condition is not overly restrictive in applied work, note that the $\mathrm{Normal}(\mu,\sigma^2)$ distribution is responsive (``to $\mu$''), and any parameterized family of distributions, $\mathcal{F}$, of real-valued random variable, $\chi$, for which $E_\mathcal{F}[\chi]$ is a non-trivial function of at least one parameter is responsive.  We refer to any indexed family of random variables, $\mathcal{P}^k$ satisfying these conditions as \textbf{responsive}.

\paragraph{The Set of Distributions Yielding Equivalence.}  As a final step before proving our main result, define the following set:
\begin{eqnarray*}
    Y_C^k & \equiv & \left\{F \in \mathcal{P}^k: \;\; \partial_t EU_F^+(0) = 0\right\} = \left\{F \in \mathcal{P}^k: \;\; \partial_t EU_F^-(0) = 0\right\},\\
    & = & \left\{F \in \mathcal{P}^k: \;\; \partial_t F(0) = 1/2 \right\}
\end{eqnarray*}
The set $Y_C^k$ contains exactly the elements of the parameterized family $\mathcal{P}^k$ under which the first order necessary condition for equivalence of accuracy-maximization and compliance-maximization (and compliance-minimization) is satisfied.  

\begin{lemma}
    The set $Y_C^k$ is a Borel subset of $\mathcal{P}^k$ in the topology induced by the $L^1$ norm.
\end{lemma}
\begin{proof}
Note that $Y_C^k$ is equivalent to the following set:
\[
\mathcal{A} := \left\{ f \in L^1(\mathbf{R}) \cap \mathcal{C}(\mathbf{R}) \, : \, f(0) = \tfrac{1}{2} \right\},
\]
so we will show that \( \mathcal{A} \) is a Borel subset of \( L^1(\mathbf{R}) \).  While the set of continuous functions is not closed in \( L^1 \) (because the limit of a sequence of continuous functions can be discontinuous), it is a Borel subset, because the set of continuous functions can be characterized as a countable union of closed sets (i.e., it is an \( F_\sigma \) set) in the \( L^1 \) norm topology:
\[
\mathcal{C}(\mathbf{R}) \cap L^1(\mathbf{R}) = \bigcup_{n=1}^\infty \bigcap_{m=1}^\infty A_{n,m},
\]
where each set \( A_{n,m} \) is defined as follows:
\[
A_{n,m} \equiv \left\{  f \in L^1(\mathbf{R} \cap \mathcal{C}(\mathbf{R}) \middle| \forall x,y \in [-n,n], |x-y|<\frac{1}{m}\Rightarrow|f(x)-f(y)|<\frac{1}{m} \right\}.
\]



Now, consider the following sequence of continuous linear functionals:
\[
\delta_n(f) := \frac{1}{2\varepsilon_n} \int_{-\varepsilon_n}^{\varepsilon_n} f(x) \, dx,
\]
where \( \varepsilon_n \to 0 \). For any \( f \in \mathcal{C}(\mathbf{R}) \), it follows that 
\[
\lim_{n\to \infty} \delta_n(f) = f(0),
\]
and then note that the set $\mathcal{A}$:
\begin{eqnarray*}
\mathcal{A} & \equiv & \bigg\{ f \in \mathcal{C}(\mathbf{R}) \cap L^1(\mathbf{R}) \, : \, f(0) = \tfrac{1}{2} \bigg\},\\
& = & \bigcap_{k=1}^\infty \bigcup_{n=1}^\infty \left\{ f \in \mathcal{C}(\mathbf{R}) \cap L^1 \, : \, \left| \delta_n(f) - \tfrac{1}{2} \right| < \tfrac{1}{k} \right\},
\end{eqnarray*}
where, for any pair of positive integers $(n,m)$,
\[
\left\{ f \in L^1 \,: \, \left| \delta_n(f) - \tfrac{1}{2} \right| < \tfrac{1}{k} \right\}
\]
is open in the $L^1$ topology because $\delta_n$ is a continuous linear functional on $L^1$. Thus, for any positive integer $k$,
\[
\bigcup_{n=1}^\infty \left\{ f \in \mathcal{C}(\mathbf{R}) \cap L^1 \, : \, \left| \delta_n(f) - \tfrac{1}{2} \right| < \tfrac{1}{k} \right\}
\]
is a countable union of open sets.  The countable intersection of open sets is a Borel set. Accordingly, $\mathcal{A}$ is a Borel subset of \( L^1(\mathbf{R}) \) with the norm topology, as was to be shown.
\end{proof}
With the necessary topological issues resolved, we now define and discuss the notion of \textit{finite shyness} prior to proving our main result.

\subsection{Finite Shyness}

Finite shyness (\cite{AndersonZame01}) is a notion of genericity (``almost always'') for infinite dimensional spaces that has intuitive measure-theoretic properties and is equivalent to being a set possessing Lebesgue measure zero in finite dimensional Euclidean spaces.\footnote{Finite shyness is an extensions of the notion of
shyness, as defined by Hunt, Sauer, and Yorke
\cite{HuntSauerYorke92}. Finite shyness is a stronger version of
shyness.}  The following formally defines the notion.

\begin{definition}
\label{Def:FiniteShyness}
Let $Q$ be a topological vector space and let $U$ be a convex subset
of $Q$ that is completely metrizable in the relative topology induced
by $Q$.  A Borel subset $E \subset U$ is \emph{finitely shy in
(}or\emph{ relative to) $U$} if there is a finite-dimensional subspace
$V \subset Q$ such that $\lambda_{V}(U+a)>0$ for some $a \in Q$ and
$\lambda_{V}(E + q) = 0$ for every $q \in Q$. An arbitrary subset $F
\subset Q$ is finitely shy in $U$ if it is contained in a finitely shy
Borel set.  If $E$ is finitely shy in $U$, then $U \setminus E$ is
referred to as \emph{finitely prevalent}.
\end{definition}
We now proceed to show that $Y_C^k$ is finitely shy in $\mathcal{P}^k$.

\subsection{The Main Result}

The following states our main result.

\begin{proposition}
\label{Pr:IntuitiveGenericityResult}
    Suppose that $\mathcal{P}^k$ is a responsive family of cost distributions for some positive integer $k$ and some convex set $X \subseteq \mathbf{R}^k$.  Then, for any admissible pair of distributions, $(G_0,G_1)$, and reward $r\neq 0$, the set $Y_C^k$ is finitely shy in $\mathcal{P}^k$.
\end{proposition}
\begin{proof}
    Fix any admissible pair of distributions, $(G_0,G_1)$, and any non-zero reward, $r\neq 0$, and any responsive family of cost distributions, $\mathcal{P}^k$, that is responsive for some positive integer $k$.  Then suppose, for the purposes of reaching a contradiction, that the set of cost distributions that induce identical compliance-maximizing and accuracy-maximizing strategies possesses positive Lebesgue measure zero.

    This would require that for \textit{each} of the $k$ parameters of the family of cost distributions (an arbitrary choice of which which we will denote by $x_1$), it is the case that, holding the other $k-1$ parameters ($\tilde{x}_{-1}=(x_2,\ldots,x_k)$ constant, there exists a set, $Z\left(x_{-1}\right)$, with non-empty interior in $\mathbf{R}$, such that $x_1\in Z$ and $x_{-1} = \tilde{x}_{-1}$ imply that 
    \[
    \partial_t EU_{F(\cdot|x_1,x_{-1})}^+(0) = 0.
    \]
    For this be the case, then for any $\varepsilon>0$, $x_1,x_1+\varepsilon\in Z$, and $x_{-1} = \tilde{x}_{-1}$, it must be the case that 
    \[
    \pi_{F(\cdot|x_1,x_{-1})}=\pi_{F(\cdot|x_1+\varepsilon,x_{-1})}=\onehalf.
    \]
    Recalling Equation \ref{Eq:EquilibriumPrevalences}, it follows that, for any $t \in \mathbf{R}$:
    \begin{eqnarray*}
    \pi_{F(\cdot|x_1,x_{-1})}^+(t) & = & F\left(r\cdot \left(G_0(t)-G_1(t)\right)|x_1,x_{-1})\right), \text{ and}\\
    \pi_{F(\cdot|x_1+\varepsilon,x_{-1})}^+(t) & = & F\left(r\cdot \left(G_0(t)-G_1(t)\right)|x_1+\varepsilon,x_{-1})\right).
    \end{eqnarray*}
    However, if $\pi_{F(\cdot|x_1,x_{-1})}^+(t)=\pi_{F(\cdot|x_1+\varepsilon,x_{-1})}^+(t)$ for any $t \in \mathbf{R}$ and each possible choice of parameter, $x_1$, then the family of distributions, $\mathcal{P}^k$, is not responsive.  Accordingly, the supposition that the cost distribution, $F(\cdot|y \in Y)$, induces identical compliance-maximizing and accuracy-maximizing strategies on a set possessing positive Lebesgue measure in $X$ leads to contradiction, implying that the set of parameters, $Y \subseteq X$, such that the corresponding cost distribution, $F(\cdot|y \in Y)$, induces identical compliance-maximizing and accuracy-maximizing strategies possesses Lebesgue measure zero in $X$.  Accordingly, Definition 2.1 in \cite{AndersonZame01} implies that the set $Y_C^k$ is finitely shy in $\mathcal{P}^k$, as was to be shown.
\end{proof}

\section{Implications and Conclusion}

Proposition \ref{Pr:IntuitiveGenericityResult} is a genericity result with potentially quite broad implications.  Specifically, they each imply that the optimal strategy to incentivize effort is ``almost certainly'' different than the optimal strategy to reward those who invested effort and \textit{not} reward those who did not invest effort (\textit{i.e.}, maximize the accuracy of the rewards).  This conclusion is not entirely surprising, as any parent knows.  However, Proposition \ref{Pr:IntuitiveGenericityResult} does not rely upon commitment problems or idiosyncratic variation of the principal's effort-maximization and/or accuracy-maximization goals.  Rather, they follow from a qualitative difference between incentives induced by accuracy-maximization versus those induced compliance-maximization.  The optimal compliance-maximization strategy is invariant to the principal's belief about the proportion of individuals who will choose to comply, but this is never true for the optimal accuracy-maximization strategy.

\bibliography{john}

@string{ eco = "Econometrica" }

@article{CrawfordSobel82,
        Author = "Vincent P. Crawford and Joel Sobel",
        Title = "{Strategic Information Transmission}",
        Journal = eco,
        Volume = "50(6)",
        Pages = "1431-1451",
        Year = "1982", }

@article{AndersonZame01,
  title={{Genericity with Infinitely Many Parameters}},
  author={Anderson, Robert M. and Zame, William R.},
  journal={Advances in Theoretical Economics},
  volume={1},
  number={1},
  pages={1--65},
  year={2001}
}

@article{HuntSauerYorke92,
  title={{Prevalence: A Translation-Invariant `Almost Every' on Infinite-Dimensional Spaces}},
  author={Hunt, B.R. and Sauer, T. and Yorke, J.A.},
  journal={Bulletin (New Series) of the
American Mathematical Society},
  volume={27},
  number={2},
  pages={217--238},
  year={1992}
}

@article{KamenicaGentzkow11,
  title={{Bayesian Persuasion}},
  author={Kamenica, Emir and Gentzkow, Matthew},
  journal={American Economic Review},
  volume={101},
  number={6},
  pages={2590--2615},
  year={2011}
}

@inproceedings{JungEtAl20,
  title={{Fair Prediction with Endogenous Behavior}},
  author={Jung, Christopher and Kannan, Sampath and Lee, Changhwa and Pai, Mallesh and Roth, Aaron and Vohra, Rakesh},
  booktitle={Proceedings of the 21st ACM Conference on Economics and Computation},
  pages={677--678},
  year={2020}
}

@misc{PennPatty24AlgEndog,
  title={{Classification, Individual Incentives, and Social Outcomes }},
author={Elizabeth Maggie Penn and John W. Patty},
howpublished={Working Paper, Emory University},
year={2024}
}

@misc{Penn25,
      title={{Optimal Classification with Outcome Performativity}}, 
      author={Elizabeth Maggie Penn},
      year={2025},
      eprint={2504.06127},
      archivePrefix={arXiv},
      primaryClass={econ.TH},
      url={https://arxiv.org/abs/2504.06127}, 
}
\end{document}